\def\isarxiv{1} 

\ifdefined\isarxiv
\documentclass[11pt]{article}

\usepackage[numbers]{natbib}

\else
\documentclass{article}
\usepackage{hyperref}
\usepackage{icml2023}
\fi

\usepackage{amsmath}
\usepackage{amsthm}
\usepackage{amssymb}
\usepackage{algorithm}
\usepackage{subfig}
\usepackage{algpseudocode}
\usepackage{graphicx}
\usepackage{grffile}
\usepackage{wrapfig,epsfig}
\usepackage{url}
\usepackage{xcolor}
\usepackage{epstopdf}

\usepackage{bbm}
\usepackage{dsfont}

\allowdisplaybreaks

\ifdefined\isarxiv

\usepackage{tikz}
\usepackage{hyperref}  
\hypersetup{colorlinks=true,citecolor=blue,linkcolor=blue} 
\usetikzlibrary{arrows}
\usepackage[margin=1in]{geometry}

\else

\usepackage{microtype}

\fi

\newtheorem{theorem}{Theorem}[section]
\newtheorem{lemma}[theorem]{Lemma}

\newtheorem{corollary}[theorem]{Corollary}

\newtheorem{assumption}[theorem]{Assumption}

\newtheorem{claim}[theorem]{Claim}

\theoremstyle{definition}
\newtheorem{definition}[theorem]{Definition}
\newtheorem{remark}[theorem]{Remark}

\newcommand{\wt}{\widetilde}
\newcommand{\ov}{\overline}
\newcommand{\N}{\mathcal{N}}
\newcommand{\R}{\mathbb{R}}

\newcommand{\eps}{\epsilon}

\DeclareMathOperator*{\E}{{\mathbb{E}}}

\DeclareMathOperator*{\Z}{\mathbb{Z}}

\DeclareMathOperator{\poly}{poly}

\DeclareMathOperator{\nnz}{nnz}

\DeclareMathOperator{\vect}{vec}

\makeatletter
\newcommand*{\RN}[1]{\expandafter\@slowromancap\romannumeral #1@}
\makeatother

\usepackage{lineno}

\begin{document}

\ifdefined\isarxiv

\date{}

\title{A Nearly-Optimal Bound for Fast Regression with $\ell_\infty$ Guarantee}
\author{
Zhao Song\thanks{\texttt{zsong@adobe.com}. Adobe Research.}
\and 
Mingquan Ye\thanks{\texttt{mye9@uic.edu}. University of Illinois at Chicago.} 
\and
Junze Yin\thanks{\texttt{junze@bu.edu}. Boston University.}
\and
Lichen Zhang\thanks{\texttt{lichenz@mit.edu}. MIT.}
}

\else


\twocolumn[

\icmltitle{A Nearly-Optimal Bound for Fast Regression with $\ell_\infty$ Guarantee}


\icmlsetsymbol{equal}{*}

\begin{icmlauthorlist}
\icmlauthor{Aeiau Zzzz}{equal,to}
\icmlauthor{Bauiu C.~Yyyy}{equal,to,goo}
\icmlauthor{Cieua Vvvvv}{goo}
\icmlauthor{Iaesut Saoeu}{ed}
\icmlauthor{Fiuea Rrrr}{to}
\icmlauthor{Tateu H.~Yasehe}{ed,to,goo}
\icmlauthor{Aaoeu Iasoh}{goo}
\icmlauthor{Buiui Eueu}{ed}
\icmlauthor{Aeuia Zzzz}{ed}
\icmlauthor{Bieea C.~Yyyy}{to,goo}
\icmlauthor{Teoau Xxxx}{ed}\label{eq:335_2}
\icmlauthor{Eee Pppp}{ed}
\end{icmlauthorlist}

\icmlaffiliation{to}{Department of Computation, University of Torontoland, Torontoland, Canada}
\icmlaffiliation{goo}{Googol ShallowMind, New London, Michigan, USA}
\icmlaffiliation{ed}{School of Computation, University of Edenborrow, Edenborrow, United Kingdom}

\icmlcorrespondingauthor{Cieua Vvvvv}{c.vvvvv@googol.com}
\icmlcorrespondingauthor{Eee Pppp}{ep@eden.co.uk}

\icmlkeywords{Machine Learning, ICML}

\vskip 0.3in
]

\printAffiliationsAndNotice{\icmlEqualContribution} 
\fi

\ifdefined\isarxiv
\begin{titlepage}
  \maketitle
  \begin{abstract}

Given a matrix $A\in \R^{n\times d}$ and a vector $b\in \R^n$, we consider the regression problem with $\ell_\infty$ guarantees: finding a vector $x'\in \R^d$ such that
\begin{align*}
    \|x'-x^*\|_\infty \leq & ~ \frac{\epsilon}{\sqrt{d}}\cdot \|Ax^*-b\|_2\cdot \|A^\dagger\|,
\end{align*}
where $x^*=\arg\min_{x\in \R^d}\|Ax-b\|_2$. One popular approach for solving such $\ell_2$ regression problem is via sketching: picking a structured random matrix $S\in \R^{m\times n}$ with $m\ll n$ and $SA$ can be quickly computed, solve the ``sketched'' regression problem $\arg\min_{x\in \R^d} \|SAx-Sb\|_2$. 

In this paper, we show that in order to obtain such $\ell_\infty$ guarantee for $\ell_2$ regression, one has to use sketching matrices that are \emph{dense}. To the best of our knowledge, this is the first user case in which dense sketching matrices are necessary. On the algorithmic side, we prove that there exists a distribution of dense sketching matrices with $m=\epsilon^{-2}d\log^3(n/\delta)$ such that solving the sketched regression problem gives the $\ell_\infty$ guarantee, with probability at least $1-\delta$. Moreover, the matrix $SA$ can be computed in time $O(nd\log n)$. Our row count is nearly-optimal up to logarithmic factors, and significantly improves the result in~\cite{psw17}, in which a super-linear in $d$ rows,  $m=\Omega(\epsilon^{-2}d^{1+\gamma})$ for $\gamma=\Theta(\sqrt{\frac{\log\log n}{\log d}})$ is required.

Moreover, we develop a novel analytical framework for $\ell_\infty$ guarantee regression that utilizes the \emph{Oblivious Coordinate-wise Embedding} (OCE) property introduced in~\cite{sy21}. Our analysis is much simpler and more general than that of~\cite{psw17}. Leveraging this framework, we extend the $\ell_\infty$ guarantee regression result to dense sketching matrices for computing the fast tensor product of vectors.

  \end{abstract}
  \thispagestyle{empty}
\end{titlepage}

{\hypersetup{linkcolor=black}
\tableofcontents
}
\newpage

\else

\begin{abstract}

Given a matrix $A\in \R^{n\times d}$ and a vector $b\in \R^n$, we consider the regression problem with $\ell_\infty$ guarantees: finding a vector $x'\in \R^d$ such that
\begin{align*}
    \|x'-x^*\|_\infty \leq & ~ \frac{\epsilon}{\sqrt{d}}\cdot \|Ax^*-b\|_2\cdot \|A^\dagger\|,
\end{align*}
where $x^*=\arg\min_{x\in \R^d}\|Ax-b\|_2$. One popular approach for solving such $\ell_2$ regression problem is via sketching: picking a structured random matrix $S\in \R^{m\times n}$ with $m\ll n$ and $SA$ can be quickly computed, solve the ``sketched'' regression problem $\arg\min_{x\in \R^d} \|SAx-Sb\|_2$. 

In this paper, we show that in order to obtain such $\ell_\infty$ guarantee for $\ell_2$ regression, one has to use sketching matrices that are \emph{dense}. To the best of our knowledge, this is the first user case in which dense sketching matrices are necessary. On the algorithmic side, we prove that there exists a distribution of dense sketching matrices with $m=\epsilon^{-2}d\log^3(n/\delta)$ such that solving the sketched regression problem gives the $\ell_\infty$ guarantee, with probability at least $1-\delta$. Moreover, the matrix $SA$ can be computed in time $O(nd\log n)$. Our row count is nearly-optimal up to logarithmic factors, and significantly improves the result in~\cite{psw17}, in which a super-linear in $d$ rows,  $m=\Omega(\epsilon^{-2}d^{1+\gamma})$ for $\gamma=\Theta(\sqrt{\frac{\log\log n}{\log d}})$ is required.

Moreover, we develop a novel analytical framework for $\ell_\infty$ guarantee regression that utilizes the \emph{Oblivious Coordinate-wise Embedding} (OCE) property introduced in~\cite{sy21}. Our analysis is much simpler and more general than that of~\cite{psw17}. Leveraging this framework, we extend the $\ell_\infty$ guarantee regression result to dense sketching matrices for computing the fast tensor product of vectors.
\end{abstract}

\fi

\section{Introduction}

Linear regression, or $\ell_2$ least-square problem is ubiquitous in numerical linear algebra, scientific computing and machine learning. Given a tall skinny matrix $A\in \R^{n\times d}$ and a label vector $b\in \R^n$, the goal is to (approximately) compute an optimal solution $x'$ that minimizes the $\ell_2$ loss $\|Ax-b\|_2$. For the regime where $n\gg d$, sketching is a popular approach to obtain an approximate solution quickly~\cite{cw13,cswz23}: the idea is to pick a random matrix $S\in \R^{m\times n}$ from carefully-designed distributions, so that 1).\ $S$ can be efficiently applied to $A$ and 2).\ the row count $m\ll n$. Given these two guarantees, one can then solve the ``sketched'' regression problem:
\begin{align*}
    \arg\min_{x\in \R^d}\|SAx-Sb\|_2,
\end{align*}
and obtain a vector $x'$ such that $\|Ax'-b\|_2= (1\pm\epsilon)\cdot {\rm OPT}$, where OPT denotes the optimal $\ell_2$ discrepancy between vectors in column space of $A$ and $b$. Recent advances in sketching~\cite{cswz23} show that one can design matrix $S$ with $m=O(\epsilon^{-2}d\log^2(n/\delta))$ and the sketched regression can be solved in time $O(\nnz(A)+d^\omega+\epsilon^{-2}d^2 \poly\log(n,d,1/\epsilon,1/\delta))$ where $\nnz(A)$ denotes the number of nonzero entries of $A$ and $\delta$ is the failure probability.

Unfortunately, modern machine learning emphasizes more and more on large, complex, and nonlinear models such as deep neural networks, thus linear regression becomes less appealing as a \emph{model}. However, it is still a very important \emph{subroutine} in many deep learning and optimization frameworks, especially second-order method for training neural networks~\cite{bpsw21,szz21} or convex optimization~\cite{ls14,lsz19,sy21,jswz21}. In these applications, one typically seeks \emph{forward error} guarantee, i.e., how close is the approximate solution $x'$ to the optimal solution $x^*$. A prominent example is Newton's method: given the (possibly implicit) Hessian matrix $H\in \R^{d\times d}$ and the gradient $g\in \R^d$, one wants to compute $H^{-1}g$. A common approach is to solve the regression $\arg\min_{x\in\R^d} \|Hx-g \|_2$, in which one wants $\|x-H^{-1}g\|_2$ or even $\|x-H^{-1}g\|_\infty$ to be small. When the matrix $S$ satisfies the so-called Oblivious Subspace Embedding (OSE) property~\cite{s06}, one can show that the approximate solution $x'$ is close to $x^*$ in the $\ell_2$ sense:
\begin{align}\label{eq:ose_forward}
    \|x'-x^*\|_2 \leq & ~ O(\epsilon)\cdot \|Ax^*-b\|_2\cdot \|A^\dagger\|.
\end{align}

Unfortunately, $\ell_2$-closeness cannot characterize how good $x'$ approximates $x^*$, as $x^*$ can have a good spread of $\ell_2$ mass over all coordinates while $x'$ concentrates its mass over a few coordinates. Formally speaking, let $a\in \R^d$ be a fixed vector, then one can measure how far $\langle a, x'\rangle$ deviates from $\langle a, x^*\rangle$ via Eq.~\eqref{eq:ose_forward}:
\begin{align*}
    |\langle a, x^*\rangle-\langle a, x'\rangle | 
    = & ~ |\langle a, x'-x^*  \rangle| \\
    \leq & ~ \|a\|_2 \|x'-x^*\|_2 \\
    \leq & ~ O(\epsilon)\cdot \|a\|_2\cdot \|Ax^*-b\|_2 \cdot \|A^\dagger\|.
\end{align*}

This bound is clearly too loose, as one would expect the deviation on a random direction is only $\frac{1}{\sqrt{d}}$ factor of the $\ell_2$ discrepancy.~\cite{psw17} shows that this intuition is indeed true when $S$ is picked as the subsampled randomized Hadamard transform ({\sf SRHT})~\cite{ldfu13}: \footnote{We will later refer the following property as $\ell_\infty$ guarantee.}
\begin{align}\label{eq:l_infty}
    |\langle a, x^*\rangle-\langle a, x'\rangle | \lesssim \frac{\epsilon}{\sqrt d} \|a\|_2  \|Ax^*-b \|_2 \|A^\dagger\|.
\end{align}

However, their analysis is not tight as they require a row count $m=\Omega(\epsilon^{-2}d^{1+\gamma})$ for $\gamma=\Theta(\sqrt{\frac{\log\log n}{\log d}})$. Such a row count is super-linear in $d$ as long as $n\leq \exp(d)$ and therefore is worse than the required row count for $S$ to be a subspace embedding, in which only $m=\epsilon^{-2}d\log^2 n$ rows are required for constant success probability. In contrast, for random Gaussian matrices, the $\ell_\infty$ guarantee only requires nearly linear in $d$ rows. In addition to their sub-optimal row count, the~\cite{psw17} analysis is also complicated: let $U\in \R^{n\times d}$ be an orthonormal basis of $A$,~\cite{psw17} has to analyze the higher moment of matrix $I_d-U^\top S^\top SU$. This makes their analysis particularly hard to generalize to other dense sketching matrices beyond {\sf SRHT}.

In this work, we present a novel framework for analyzing the $\ell_\infty$ guarantee induced by {\sf SRHT} and more generally, a large class of \emph{dense} sketching matrices. Our analysis is arguably much simpler than~\cite{psw17}, and it exposes the fundamental structure of sketching matrices that provides $\ell_\infty$ guarantee: if any two columns of the sketching matrix have a small inner product with high probability, then $\ell_\infty$ guarantee can be preserved. We then prove that the small pairwise column inner product is also closely related to the \emph{Oblivious Coordinate-wise Embedding} ({\sf OCE}) property introduced in~\cite{sy21}. More concretely, for any two fixed vectors $g, h\in \R^n$, we say the sketching matrix is $(\beta,\delta,n)$-{\sf OCE} if $|\langle Sg, Sh\rangle-\langle g,h\rangle|\leq \frac{\beta}{\sqrt{m}}\cdot \|g\|_2 \|h\|_2$ holds with probability at least $1-\delta$. This property has previously been leveraged for approximating matrix-vector product between a dynamically-changing projection matrix and an online sequence of vectors for the fast linear program and empirical risk minimization algorithms~\cite{lsz19,jswz21,sy21,qszz23} as these algorithms need $\ell_\infty$ bound on the matrix-vector product. One common theme shared by those applications and $\ell_\infty$ guarantee is to use \emph{dense} sketching matrices, such as random Gaussian, the Alon-Matias-Szegedy sketch (AMS,~\cite{ams96}) or {\sf SRHT}. This is in drastic contrast with the trending direction for using sparse matrices such as Count Sketch~\cite{ccf02,cw13} and OSNAP~\cite{nn13,c16}, as they can be applied in (nearly) input sparsity time.

In recent years, sketches that can be applied to the tensor product of matrices/vectors have gained popularity~\cite{anw14,swz16,dssw18,swz19,djs+19,akk+20,wz20,swyz21,wz22,sxz22,sxyz22,rsz22} as they can speed up optimization tasks and large-scale kernel learning. We show that dense sketches for degree-2 tensors also provide $\ell_\infty$ guarantee.

\begin{theorem}[Nearly-optimal bound for dense sketching matrices]\label{thm:standard_version}
Suppose $n\leq \exp(d)$ and matrix $A \in \R^{n \times d}$ and vector $b \in \R^{n}$ are given.
Let $S \in \R^{m \times n}$ be a subsampled randomized Hadamard transform matrix {\sf SRHT} with $m= O(  \epsilon^{-2} d \log^3(n/\delta) )$ rows. 

For any fixed vector $a \in \R^d$,
\begin{align*}
| \langle a, x^* \rangle - \langle a , x' \rangle | \lesssim & ~  \frac{ \epsilon } { \sqrt{d} } \cdot\| a \|_2 \cdot \| A x^* - b \|_2 \cdot \| A^\dagger \|
\end{align*}
with probability $1-\delta$, where $
    x' = \arg\min_{x \in \R^d} \| SA x - S b \|_2
$,
$
    x^* = \arg\min_{x \in \R^d} \| A x - b \|_2$.
\end{theorem}

\begin{remark}\label{rem:row_count}
The row count $m=\epsilon^{-2}d\log^3 n$ is nearly-optimal up to logarithmic factors, as the row count for $S$ being an OSE is $m=\epsilon^{-2}d\log^2 n$ for constant success probability. In comparison,~\cite{psw17} requires $m=\epsilon^{-2}d^{1+\gamma}$ rows for $\gamma=\Theta(\sqrt{\frac{\log\log n}{\log d}})$ which is only nearly-linear in $d$ if $n>\exp(d)$. In most applications, we concern about $n=\poly(d)$, meaning that their row count is worse than ours in almost all meaningful scenarios.

The row count and guarantee obtained in Theorem~\ref{thm:standard_version} extend beyond {\sf SRHT}; in fact, for a range of \emph{dense} sketching matrices including random Gaussian, AMS sketch~\cite{ams96}, {\sf SRHT} and subsampled randomized Circulant Transform (see Definition~\ref{def:CGT}) ({\sf SRCT}). This is because our argument is a structural condition that can be satisfied by various dense sketches. 

Our result can also be generalized to degree-2 Kronecker product regression, see Theorem~\ref{thm:tensor_version}.
\end{remark}

\paragraph{Roadmap.}

In Section~\ref{sec:preliminary}, we introduce the notations that we use and explain the key definitions and properties to support the framework for $\ell_\infty$ guarantee regression. In Section~\ref{sec:oce}, we introduce our framework by presenting a sufficient condition for a sketching matrix to give a good $\ell_\infty$ guarantee. In Section~\ref{sec:put_together}, we provide a proof for our main theorem by putting everything together. Finally, in Section~\ref{sec:conclusion}, we summarize the main findings of this paper and through comparing with previous work.  

\section{Preliminary}
\label{sec:preliminary}

For a positive integer, we define $[n]:=\{1,2,\cdots, n\}$. For a vector $x \in \R^n$, we define $\| x \|_2 : = ( \sum_{i=1}^n x_i^2 )^{1/2}$ and $\| x \|_{\infty}:= \max_{i \in [n]} |x_i|$. For a matrix $A$, we define $\| A \| := \sup_{x} \| A x \|_2 / \| x \|_2$ to be the spectral norm of $A$. We use $\| A \|_F : =\sum_{i,j} A_{i,j}^2$ to be the Frobenius norm of $A$. In general, we have the following property for spectral norm, $\| A B \| \leq \| A \| \cdot \| B \|$. 
We use $A^\dagger$ to denote the Moore-Penrose pseudoinverse of $m \times n$ matrix $A$ which if $A = U \Sigma V^\top$ is its SVD (where $U \in \R^{m \times n}$, $\Sigma \in \R^{n \times n}$ and $V\in \R^{n \times n}$ for $m \geq n$), is given by $A^\dagger = V \Sigma^{-1} U^\top$. 

We use $\E[\cdot]$ to denote the expectation, and $\Pr[\cdot]$ to denote the probability. 
For a distribution $D$ and a random variable $x$, we use $x \sim D$ to denote that we draw a random variable from the distribution $D$. 
We use ${\cal N}(\mu, \sigma^2)$ to denote a Gaussian distribution with mean $\mu$ and variance $\sigma^2$. 
We say a random variable $x$ is Rademacher random variables if $\Pr[x=1]=1/2$ and $\Pr[x=-1]=1/2$. We also call it a sign random variable. 

In addition to $O(\cdot)$ notation, for two functions $f,g$, we use the shorthand $f \lesssim g$ (resp. $\gtrsim$) to indicate that $f \leq C g$ (resp. $\geq$) for an absolute constant $C$. We use $f \eqsim g$ to mean $c f \leq g \leq C f$ for constants $c > 0$ and $C> 0$. 
For two matrices $A\in \R^{n_1\times d_1}$, $B\in \R^{n_2\times d_2}$, we use $A\otimes B \in \R^{n_1n_2\times d_1d_2}$ to denote the Kronecker product, i.e., the $(i_1,i_2)$-th row and $(j_1,j_2)$-th column of $A\times B$ is $A_{i_1,j_1}B_{i_2,j_2}$. For two vectors $x\in \R^m, y\in \R^n$, we use $x\otimes y\in \R^{mn}$ to denote the tensor product of vectors, in which $x\otimes y=\vect(xy^\top)$.

\subsection{Oblivious subspace embedding and coordinate-wise embedding}

\begin{definition}[Oblivious subspace embedding~\cite{s06}]\label{def:ose}
We define $(\epsilon,\delta,d,n)$-Oolivious subspace embedding (\textsf{OSE}) as follows: Suppose $\Pi$ is a distribution on $m \times n$ matrices $S$, where $m$ is a function of $n, d, \epsilon$, and $\delta$. Suppose that with probability at least $1 - \delta$, for
any fixed $n \times d$ orthonormal basis $U$, a matrix $S$ drawn from the distribution $\Pi$ has the property that the singular values of $SU$ lie in the range $[1-\epsilon,1+\epsilon]$.
\end{definition}

The oblivious coordinate-wise embedding ({\sf OCE}) is implicitly used in \cite{lsz19,jswz21} and formally introduced in \cite{sy21}. 
\begin{definition}[$(\alpha, \beta, \delta)$--coordinate wise embedding~\cite{sy21}]\label{def:full_oce}
    We say a random matrix $S \in \R^{m \times n}$ satisfying  $(\alpha, \beta, \delta)$--coordinate wise embedding if
    \begin{align*}
        & ~ 1. \E_{S\sim \Pi} [g^\top S^\top Sh] = g^\top h,\\
        & ~ 2. \E_{S\sim \Pi} [(g^\top S^\top Sh)^2] \leq (g^\top h)^2 + \frac{\alpha}{m} \|g\|_2^2 \|h\|_2^2,\\
        & ~ 3. \Pr_{S\sim \Pi} [|g^\top S^\top Sh - g^\top h| \geq \frac{\beta}{\sqrt{m}} \|g\|_2 \|h\|_2] \leq \delta .
    \end{align*}
\end{definition}

In this paper, we mainly use the property 3 of Definition~\ref{def:full_oce}. For convenient, we redefine {\sf OCE} as follows:

\begin{definition}[{\sf OCE}]\label{def:oce}
Let $\beta \geq 1$ and $\delta \in (0,0.1)$. We say a randomized matrix $S \in \R^{m \times n}$ satisfy $(\beta,\delta,n)$-{\sf OCE}, if 
\begin{align*}
\Pr_{S \sim \Pi}[ | g^\top S^\top S h - g^\top h| \geq \frac{\beta}{\sqrt{m}} \| g \|_2  \| h \|_2  ] \leq \delta
\end{align*}
and the distribution $\Pi$ is oblivious to any fixed vectors $g$ and $h$.
\end{definition}

\subsection{Sketching matrices}

In this paper, we concern a list of dense sketching matrices. 

\begin{definition}[Random Gaussian matrix, folklore]\label{def:random_Gaussian}
    We say $S \in \R^{m\times n}$ is a random Gaussian matrix if all entries are sampled from $\N(0, 1/m)$ independently.
\end{definition}

\begin{definition}[AMS sketch matrix, \cite{ams96}]\label{def:AMS}
    Let $h_1, h_2, \cdots , h_m$ be $m$ random hash functions picking from a 4--wise independent hash family $\mathcal{H} = \{h:[n] \to \{-\frac{1}{m}, + \frac{1}{m}\}\}$. Then $S \in \R^{m \times n}$ is an AMS sketch matrix if we set $S_{i,j} = h_i(j)$.
\end{definition}

The following sketching matrices can utilize fast Fourier Transform (FFT) for efficient application to matrices.
\begin{definition}[Subsampled randomized Hadamard transform (\textsf{SRHT}) \cite{ldfu13,swyz21}]\label{def:srht}
    The \textsf{SRHT} matrix $S\in\R^{m \times n}$ is defined as $S:=\frac{1}{\sqrt{m}} P H D$, where each row of matrix $P \in \{0, 1\}^{m \times n}$ contains exactly one $1$ at a random position, $H$ is the $n \times n$ Hadamard matrix, and $D$ is a $n \times n$ diagonal matrix with each diagonal entry being a value in $\{-1, +1\}$ with equal probability. 
\end{definition}

\begin{remark}\label{rem:FFT}
    Using the fast Fourier transform (FFT), $S$ can be applied to a vector in time $O(n\log{n})$. 
\end{remark}

\begin{definition}[Tensor subsampled randomized Hadamard transform (\textsf{TensorSRHT}) \cite{swyz21}]\label{def:tensor_srht}
The \textsf{TensorSRHT} $S: \R^n \times \R^n \to \R^m$ is defined as $S := \frac{1}{\sqrt{m}} P \cdot (HD_1 \otimes HD_2)$, where each row of $P \in \{0, 1\}^{m \times n^2}$ contains only one $1$ at a random coordinate and one can view $P$ as a sampling matrix. $H$ is a $n \times n$ Hadamard matrix, and $D_1$, $D_2$ are two $n \times n$ independent diagonal matrices with diagonals that are each independently set to be a Rademacher random variable (uniform in $\{-1, 1\}$).  
\end{definition}

\begin{remark}\label{rem:computed_in_time}
    By leveraging the FFT algorithm in the sketch space, $S(x\otimes y)$ can be computed in time $O(n\log{n} + m)$. 
\end{remark}

To store and generate a Hadamard matrix is expensive, we consider a cheaper and space-efficient way to generate an FFT matrix via circulant transform. 

\begin{definition}[Circulant matrix]\label{def:circulant_matrix}
    A circulant matrix is an $n \times n$ matrix, where $n\in\mathbb{N}$, whose row vectors consist of the same element, and compared to the preceding row vector, each row vector is rotated one element to the right.
\end{definition}

\begin{definition}[Subsampled randomized circulant transform ({\sf SRCT})]
\label{def:CGT}
Let $x \in \R^{n}$ be a random vector, whose elements are i.i.d. Rademacher random variables. 

Also, let $P \in \R^{m \times n}$ be a random matrix in which each row contains a 1 at a uniformly distributed coordinate and zeros elsewhere. 

Let $G \in \R^{n \times n}$ be a circulant matrix (see Definition~\ref{def:circulant_matrix}) generated by $x$ and $D \in \R^{n \times n}$ be a diagonal matrix whose diagonal elements are i.i.d. Rademacher random variables. 

Then, the subsampled randomized circulant transform is defined as follows:
$
    S :=  ~ \frac{1}{\sqrt{m}}PGD.
$
\end{definition}

\begin{definition}[Tensor subsampled randomized circulant transform ({\sf TensorSRCT})]
\label{def:Tensor_CGT}
Let $x \in \R^{n}$ be a random vector, whose elements are i.i.d. Rademacher random variables. 

Also, let $P \in \R^{m \times n^{2}}$ be a random matrix in which each row contains a 1 at a uniformly distributed coordinate and zeros elsewhere.

Let $G \in \R^{n \times n}$ be a circulant matrix (see Definition~\ref{def:circulant_matrix}) generated by $x$. 

Let $D_{1} \in \R^{n \times n}$ and $D_{2} \in \R^{n \times n}$ be two independent diagonal matrices whose diagonal elements are i.i.d. Rademacher random variables.

Then, the tensor circulant transform $T: \R^{n} \times \R^{n} \rightarrow \R^{m}$ is defined as follows:
$
    T := P \cdot (G D_{1} \otimes G D_{2}).
$
\end{definition}

\begin{remark}
Similar to {\sf SRHT}, we can utilize the fast Fourier transform with circulant matrix. {\sf SRCT} can be applied to a vector of length $n$ in $O(n\log n)$ time, and {\sf TensorSRCT} can be applied to $x\otimes y$ in $O(n\log n+m)$ time.
\end{remark}

\subsection{{\sf OSE} property of dense sketches}

An important condition for sketch-and-solve regressions is {\sf OSE}. We focus particularly on {\sf SRHT}, {\sf SRCT}, and their tensor variants.

\begin{lemma}[Lemma 2.11 in~\cite{swyz21}]\label{lem:srht_implies_ose}
    Let $S$ be an {\sf SRHT} matrix defined in Definition~\ref{def:srht}. If $m=O( \eps^{-2}d \log^2(nd / \delta))$, then $S$ is an $(\eps, \delta, d, n)$-{\sf OSE}. 
\end{lemma}

\begin{lemma}[Lemma 2.12 in~\cite{swyz21}]\label{lem:tensor_srht_implies_ose}
    Let $S$ be a {\sf TensorSRHT} matrix defined in Definition~\ref{def:tensor_srht}. If $m=O(\eps^{-2}d \log^3(nd / \eps\delta) )$, then $S$ is an $(\eps, \delta, d, n^2)$-{\sf OSE} for degree-$2$ tensors. 
\end{lemma}

{\sf SRCT} requires more row count than {\sf SRHT} due to the Gram $G^\top G$ is only $I_n$ in expectation.

\begin{lemma}[Informal version of Corollary~\ref{cor:srct_imply_ose}]
Let $S$ be an {\sf SRCT} matrix defined in Definition~\ref{def:CGT}. If $m=O(\epsilon^{-2}d^2 \log^2(nd/\delta))$, then $S$ is an $(\eps,\delta,d,n)$-{\sf OSE}.
\end{lemma}

\begin{lemma}[Informal version of Corollary~\ref{cor:tensor_srct_ose_1}]
Let $S$ be an {\sf TensorSRCT} matrix defined in Definition~\ref{def:Tensor_CGT}. If $m=O(\epsilon^{-2}d^2 \log^3(nd/\delta))$, then $S$ is an $(\eps,\delta,d,n^2)$-{\sf OSE}.
\end{lemma}

\subsection{Probability tools}

\begin{lemma}[Khintchine's inequality \cite{k23}]\label{lem:khintchine}
    Let $\sigma_1, \cdots, \sigma_n$ be i.i.d. Rademacher random variables and $z_1, \cdots, z_n$ be real numbers. Then, there exists constants $C, C'>0$ such that 
    \begin{align*}
        \Pr[|\sum_{i=1}^{n} z_i \sigma_i| \ge C t \|z\|_2] \le \exp( -C' t^2) .
    \end{align*}
\end{lemma}

\begin{lemma}[Hoeffding bound, \cite{h63}]
\label{lem:hoeffding}

Let $Z_1,\cdots,Z_n$ be independent, zero-mean random variables with $Z_i\in [\alpha_i,\beta_i]$. Then,
\begin{align*}
    \Pr[|\sum_{i=1}^n Z_i|>t] \leq & ~ 2\exp(-\frac{t^2}{2\sum_{i=1}^n (\beta_i-\alpha_i)^2}).
\end{align*}
\end{lemma}

\begin{lemma}[Lemma 1 on page 1325 of Laurent and Massart \cite{lm00}
]\label{lem:lm}
    Let $X \sim \mathcal{X}_k^2$ be a chi-squared distributed random variable with $k$ degrees of freedom. Each one has zero means and $\sigma^2$ variance. Then,
    \begin{align*}
        \Pr[X - k\sigma^2 \geq (2\sqrt{kt} + 2t) \sigma^2]
        \leq & ~ \exp{(-t)}\\
        \Pr[k\sigma^2 - X \geq 2\sqrt{kt}\sigma^2]
        \leq & ~ \exp{(-t)}
    \end{align*}
\end{lemma}

\begin{lemma}[Hanson-Wright inequality \cite{hw71}]\label{lem:hason_wright}
    Let $x \in \R^n$ denote a random vector with independent entries $x_i$ with $\E[x_i] = 0$ and $|x_i| \leq K$. Let $A$ be an $n \times n$ matrix. Then, for every $t \geq 0$,
    \begin{align*}
        & ~ \Pr[|x^\top Ax - \E[x^\top Ax]| > t] \\
        \leq & ~ 2 \cdot \exp{(-c\min\{ t^2/ (K^4 \|A\|_F^2), t/(K^2 \|A\|) \}  ) }
    \end{align*}
\end{lemma}

\begin{lemma}[Matrix Chernoff bound, Theorem 2.2 in \cite{t10}]\label{lem:chernoff_bound}
    Let $X$ be a finite set of positive-semidefinite matrices with dimension $d \times d$. Suppose that
    $\max_{X\in\mathcal{X}} \lambda_{\max} (X) \leq B$. Sample $\{X_1, \cdots , X_n\}$ uniformly at random from $\mathcal{X}$ without replacement. We define $\mu_{\min}$ and $\mu_{\max}$ as follows: $\mu_{\min} := n \cdot \lambda_{\min} (\E_{X\sim \mathcal{X}} [X])$ and $\mu_{\max} := n \cdot \lambda_{\max} (\E_{X\sim \mathcal{X}} [X])$. Then,
    \begin{align*}
        \Pr[\lambda_{\min}(\sum_{i = 1}^n X_i) 
        \leq & ~ (1 - \delta)\mu_{\min}]\leq d\cdot \exp{(-\delta^2\mu_{\min}/ B)} 
    \end{align*}
    for $\delta\in[0,1)$,
    \begin{align*}
        & ~ \Pr[\lambda_{\max}(\sum_{i = 1}^n X_i) \leq  (1 + \delta)\mu_{\max}]\\
        \leq & ~ d\cdot \exp{(-\delta^2\mu_{\max}/ (4B))} 
    \end{align*}
    for $\delta\geq 0$.
\end{lemma}

\section{\texorpdfstring{$\ell_\infty$}{} guarantee via {\sf OCE}}
\label{sec:oce}

In this section, we present a sufficient condition for a sketching matrix to give good $\ell_\infty$ guarantee: given a pair of fixed vectors $g, h$ such that $g^\top h=0$, if the sketching matrix approximately preserves the inner product with high probability, then it gives good $\ell_\infty$ guarantee for regression.

\begin{lemma}[Core lemma]\label{lem:core_lemma}
Let $A \in \R^{n \times d}$ be a fixed matrix. Let $U \in \R^{n \times d}$ denote the orthonormal basis of $A$. Let $S\in \R^{m\times n}$ be a sketching matrix that satisfies two properties 
    \begin{itemize}
    \item $S$ is an $(0.1,\delta,d,n)$-$\mathsf{OSE}$ (with $\delta \in (0,0.1)$, Definition~\ref{def:ose}).
    \item $S$ is an $(\beta,\delta,n)$-$\mathsf{OCE}$ (with $\beta \geq 1$ and $\delta \in (0,0.1)$, Definition~\ref{def:oce}).
    \end{itemize}

    For any fixed vectors $a \in \R^d$ and $b \in \R^n$ with
     $U^\top b = 0$, we have
    \begin{align*}
         |a^\top (SA)^\dag Sb|\lesssim & ~ \frac{\beta}{\sqrt{m}}\cdot \|a\|_2 \cdot \|b\|_2 \cdot \|\Sigma^{-1}\|
    \end{align*}
    holds with probability at least $1-\delta$.
\end{lemma}

\begin{proof}
    With probability 1, the matrix $SA \in \R^{m \times d}$ has linearly independent columns.
    
    Therefore, $(SA)^\dagger \in \R^{d \times m}$ is
    \begin{align*}
    (SA)^{\dagger} 
        = & ~ (A^\top S^\top SA)^{-1} A^\top S^\top\\
        = & ~ (V\Sigma U^\top S^\top S U\Sigma V^\top)^{-1} V\Sigma U^\top S^\top\\
        = & ~ (V^\top)^{-1} \Sigma^{-1} (U^\top S^\top S U)^{-1} \Sigma^{-1} V^{-1} V\Sigma U^\top S^\top\\
        = & ~ V\Sigma^{-1} (U^\top S^\top S U)^{-1} U^\top S^\top,
    \end{align*}
    where the first step follows from $SA \in \R^{m \times d}$ has full rank, the second step follows from SVD on $A \in \R^{n \times d}$, the third step follows from $(AB)^{-1} = B^{-1} A^{-1}$, and the last step follows from the fact that $V$ is orthogonal based on the property of SVD.

For convenience, we define $x$ as follows:
    \begin{align*}
        x := a^\top V\Sigma^{-1} (U^\top S^\top S U)^{-1} U^\top S^\top Sb.
    \end{align*}
In the next few paragraphs, we will explain how to upper bound $|x|$ with high probability.

    Since $S$ is a $(0.1,\delta,d, n)$-{\sf OSE} (Definition~\ref{def:ose}), we know
    \begin{align*}
    \Pr[ \| I - U^\top S^\top S U \| \leq 0.1 ] \geq 1 - \delta.
    \end{align*}

    We condition on this event. It follows that
    \begin{align*}
        & ~ \|V \Sigma^{-1} (U^\top S^\top SU)^{-1} U^\top\|\\
        = & ~ \|\Sigma^{-1} (U^\top S^\top SU)^{-1} U^\top\|\\
        \leq & ~ \|\Sigma^{-1}\| \| (U^\top S^\top SU)^{-1}\| \| U^\top \|\\
        \leq & ~ \|\Sigma^{-1}\| \cdot \frac{1}{1 - 0.1} \cdot 1\\
        = & ~ O(\|\Sigma^{-1}\|),
    \end{align*}
    where the first step follows from that $V$ is a rotation, the second step follows from sub-multiplicativity, and the third step follows from $\| I - U^\top S^\top S U \|\leq 0.1$ and that $U$ is a rotation.

    Hence, we have 
    \begin{align}\label{eq:O_Sigma_-1_a}
       & ~ \Pr[  \|a^\top V\Sigma^{-1}(U^\top S^\top SU)^{-1} U^\top \|_2 \notag \\
       = & ~ O(\|\Sigma^{-1}\| \cdot \|a\|_2) ] 
       \geq  ~ 1- \delta.
    \end{align}

    Let us define a vector $u \in \R^n$
    \begin{align*}
    u := & ~ U(U^\top S^\top SU)^{-1}\Sigma^{-1}V^\top a
    \end{align*}
 
    By the definition of {\sf OCE} (Definition~\ref{def:oce}, we have that 
    \begin{align*}
        \Pr[|u^\top S^\top Sb-u^\top b|\leq \frac{ \beta }{\sqrt{m}}\cdot \|u\|_2\|b\|_2] \leq & ~ 1-\delta,
    \end{align*}

    where $U^\top b=0$ gives us $u^\top b=0$ and $u^\top S^\top Sb=x$.
    
    Thus, the above bound translates to
    \begin{align}
        \Pr[|x|\leq C\cdot \frac{\beta}{\sqrt{m}}\cdot \|\Sigma^{-1}\|\|a\|_2\|b\|_2] \geq & ~ 1-\delta
    \end{align}
    as desired. \qedhere
\end{proof}

We are now ready to prove the $\ell_\infty$ guarantee given the inner product bound of Lemma~\ref{lem:core_lemma}.

\begin{theorem}\label{thm:full_column_rank}
Suppose $A\in \R^{n\times d}$ has full column rank and $b\in \R^n$. Let $S\in \R^{m\times n}$ be a sketching matrix satisfying conditions in Lemma~\ref{lem:core_lemma}. For any fixed vector $a\in \R^d$, we have
\begin{align*}
    |\langle a,x^*\rangle-\langle a,x'\rangle | \lesssim & ~ \frac{\epsilon}{\sqrt{d}}\cdot \|a\|_2\cdot \|Ax^*-b\|_2\cdot \|A^\dagger\|,
\end{align*}
holds with probability at least $1-\delta$, where $x^*=\arg\min_{x\in \R^d} \|Ax-b\|_2$ and $x'=\arg\min_{x\in \R^d}\|SAx-Sb\|_2$.
\end{theorem}

\begin{proof}
Since $A$ has full column rank, we have that $x^*=A^\dagger b$. Similarly, $SA$ has full column rank with probability 1, therefore $x'=(SA)^\dagger Sb$ and $(SA)^\dagger SA=I$. Thus, we have
\begin{align}
\label{eq:direction_1}
    |\langle a,x^*\rangle-\langle a,x'\rangle | = & ~ |\langle a, x^*-(SA)^\dagger Sb\rangle |\notag \\
    = & ~ |\langle a, (SA)^\dagger S(Ax^*-b)\rangle  |\notag \\
    = & ~ |\langle a, (SA)^\dagger S(AA^\dagger b-b) | \notag\\
    = & ~ |\langle ((SA)^\dagger S)^\top a, (I-UU^\top)b\rangle |
\end{align}
where $U\in \R^{n\times d}$ is an orthonormal basis for $A$. It is well-known that $I-UU^\top=U_\perp U_\perp^\top$ where $U_\perp\in \R^{n\times (n-d)}$ is the orthonormal basis for the orthogonal component of ${\rm span}(A)$. To maximize the above expression, we shall let $b\in {\rm span}(U_\perp)$ or equivalently, $U^\top b=0$. Thus, bounding Eq.~\eqref{eq:direction_1} is equivalent to consider
\begin{align*}
    |a^\top (SA)^\dagger Sb | \lesssim & ~ \frac{\beta}{\sqrt{m}}\cdot \|a\|_2\cdot \|b\|_2\cdot \|A^\dagger \|,
\end{align*}

the inequality holds with probability at least $1-2\delta$ by Lemma~\ref{lem:core_lemma}. Finally, note that since $U^\top b=0$, we have that $Ax^*=0$ and we have proved 
\begin{align*}
    |\langle a,x^*\rangle-\langle a,x'\rangle | \lesssim & ~ \frac{\epsilon}{\sqrt{d}}\cdot \|a\|_2\cdot \|Ax^*-b\|_2\cdot \|A^\dagger\|.
\end{align*}
\end{proof}

Note that we only require the ${\sf OSE}$ with $\epsilon=O(1)$ and the $\epsilon$ dependence follows from the row count of {\sf OCE}.

\subsection{High probability bound for {\sf OCE}}

In this section, we provide a unified framework for proving the high probability bound of {\sf OCE}. Our analysis utilizes the three dense sketching matrices that can all be designed as first picking a set of fresh random signs, then picking the sketching matrix according to the distribution.

We state the key assumptions on dense sketching matrices that are sufficient for {\sf OCE} property.

\begin{assumption}
\label{assumption:OCE}
Let $S\in \R^{m\times n}$ be a dense sketching matrix satisfying the following two assumptions:

\begin{itemize}
    \item Pairwise inner product bound: 
    \begin{align*}
        \Pr[\max_{i\neq j} |\langle S_{*,i}, S_{*,j}\rangle |\leq \frac{\sqrt{\log(n/\delta)}}{\sqrt{m}}] \geq & ~ 1-\delta.
    \end{align*}
    \item Column norm bound:
    \begin{align*}
        \Pr[|\|S_{*,i}\|_2^2-1 |\leq \frac{\sqrt{\log(n/\delta)}}{\sqrt{m}}] \geq & ~ 1-\delta,
    \end{align*}
    for all $i\in [n]$.
\end{itemize}
\end{assumption}

\begin{lemma}\label{lem:SRHT_AMS}
    Let $S \in \R^{m \times n}$ be a dense sketching matrix meets Assumption~\ref{assumption:OCE}. Let $h \in \R^n$ and $g \in \R^n$ be two fixed vectors. Then, the following properties hold:
    \begin{align*}
        |(g^\top S^\top Sh) - (g^\top h)| \leq & ~ \frac{\log^{1.5}(n/\delta)}{\sqrt{m}} \|g\|_2 \|h\|_2
    \end{align*}
    holds with probability at least $1-\delta$.
\end{lemma}

\begin{proof}
    We can rewrite $(g^\top S^\top Sh) - (g^\top h)$ as follows:
    \begin{align*}
       & ~ (g^\top S^\top Sh) - (g^\top h) \\
        = & ~ \sum_{i = 1}^n \sum_{j\in [n] \setminus i}^n g_ih_j \langle S_{*,i}, S_{*,j} \rangle + \sum_{i = 1}^n g_ih_i (\|S_{*,i}\|_2^2 - 1)\\
        = & ~ \underbrace{\sum_{i = 1}^n \sum_{j\in [n] \setminus i}^n g_ih_j \langle \sigma_i \overline{S}_{*,i}, \sigma_j \overline{S}_{*,j} \rangle}_{\texttt{off-diag}} \\
        + & ~ \underbrace{\sum_{i = 1}^n g_ih_i (\|S_{*,i}\|_2^2 - 1)}_{\texttt{diag}},
    \end{align*}
    where the first step follows from the fact that $\sigma_i$’s are independent Rademacher random variables and $S_{*,i} =\sigma_i  \overline{S}_{*,i} $, $\forall i \in [n]$, the second step follows from separating diagonal and off-diagonal terms. 
    
    We will focus on bounding the quantity \texttt{off-diag}, as \texttt{diag} can be handled in a rather simple fashion.

    We define matrix $A \in \R^{n\times n}$ and $B \in \R^{n\times n}$ as follows:
    \begin{align*}
        & ~ A_{i,j} := g_ih_j \cdot \langle \overline{S}_{*,i}, \overline{S}_{*,j} \rangle, & ~ \forall i \in [n],j\in [n] \\
        & ~ B_{i,j} := g_ih_j \cdot \max_{i'\not= j'}|\langle \overline{S}_{*,i'}, \overline{S}_{*,j'} \rangle|, & ~ \forall i \in [n],j\in [n].
    \end{align*}

    We define $A^\circ \in \R^{n \times n}$ to be the matrix $A \in \R^{n\times n}$ with removing diagonal entries. 
    
    By applying Hanson-Wright inequality (Lemma~\ref{lem:hason_wright}), we have
    \begin{align*}
        & ~ \Pr[|\sigma^\top A^\circ \sigma| > \tau] \\
        \leq & ~ 2 \cdot \exp{(-c\cdot \min\{ \tau^2/  \|A^\circ \|_F^2, \tau/\|A^\circ \| \}  ) }
    \end{align*}
    We can upper bound $\|A^\circ\|$ and $\|A^\circ\|_F$.
    \begin{align*}
        \|A^\circ\| 
        \leq & ~ \|A^\circ\|_F \\
        \leq & ~ \|A\|_F \\
        \leq & ~ \|B\|_F \\
        \leq & ~ \|g\|_2 \cdot \|h\|_2 \cdot \max_{i \not= j}|\langle \overline{S}_{*,i}, \overline{S}_{*,j} \rangle|, 
    \end{align*}
    where the first step follows from $\|\cdot\| \leq \|\cdot\|_F$, the second step follows from the definition of $A^\circ$, the third step follows from the definition of $A$ and $B$, and the fourth step follows from $B$ is rank 1 as $B=\max_{i\neq j}|\langle \ov S_{*,i}, \ov S_{*,j}\rangle |\cdot gh^\top$. 

    It remains to obtain a bound on $\max_{i\neq j}|\langle\ov S_{*,i}, \ov S_{*,j}\rangle|$. Note that for any column $i, j$,
    \begin{align*}
        |\langle \ov S_{*,i}, \ov S_{*,j}\rangle| = & ~ |\langle \sigma_i \ov S_{*,i}, \sigma_j \ov S_{*,j}\rangle | \\
        = & ~ |\langle S_{*,i}, S_{*,j}\rangle |,
    \end{align*}
    where the first step follows from the fact that random signs do not change the magnitude of the inner product and the second step follows from the definition of $S_{*,i}$ and $S_{*,j}$. 
    
    Since $S$ meets Assumption~\ref{assumption:OCE}, we have that with probability at least $1-\delta$,
    \begin{align*}
        \max_{i\neq j} |\langle \ov S_{*,i}, \ov S_{*,j}\rangle | \leq & ~ \frac{\sqrt{\log(n/\delta)}}{\sqrt{m}}.
    \end{align*}

    Conditioning on the above event holds, we have that 
    \begin{align*}
        & ~ \Pr[|\texttt{off-diag} |>\tau] \\
        \leq & ~ 2\cdot \exp(-c\cdot \frac{\tau}{\|g\|_2\cdot \|h\|_2\cdot \frac{\sqrt{\log(n/\delta)}}{\sqrt{m}}}),
    \end{align*}
    
    choosing $\tau = \|g\|_2 \cdot \|h\|_2 \cdot \log^{1.5}(n/\delta)/\sqrt{m}$, we can show that
    \begin{align*}
        \Pr [|\texttt{off-diag}| \geq \|g\|_2 \cdot \|h\|_2 \frac{\log^{1.5}(n/\delta)}{\sqrt{m}} ] \leq \Theta(\delta).
    \end{align*}

    To bound the term \texttt{diag}, note that due to Assumption~\ref{assumption:OCE}, we have with probability at least $1-\delta$, $|\|S_{*,i}\|_2^2-1|\leq \frac{\sqrt{\log(n/\delta)}}{\sqrt m}$. 
    
    Conditioning on this event, we have
    \begin{align*}
        |\texttt{diag}| \leq & ~ \max_{i\in [n]} |\|S_{*,i}\|_2^2-1|\cdot |g^\top h | \\
        \leq & ~ \frac{\sqrt{\log(n/\delta)}}{\sqrt m}\cdot \|g\|_2\cdot \|h\|_2,
    \end{align*}
    where the last step is by Cauchy-Schwartz. Note that $|\texttt{diag}|$ is subsumed by $|\texttt{off-diag} |$. 
    
    Union bounding over all events, we have that
    \begin{align*}
        & ~ \Pr[|g^\top S^\top Sh-g^\top h |\geq \frac{\log^{1.5}(n/\delta)} {\sqrt{m}}\cdot\|g\|_2\cdot \|h\|_2 ] \\ 
        \leq & ~\Theta(\delta). \qedhere
    \end{align*}
    
\end{proof}

\subsection{Inner product bound for {\sf SRHT} and {\sf SRCT}}

We will show that {\sf SRHT} and {\sf SRCT} satisfy Assumption~\ref{assumption:OCE}. Before proving the pairwise inner product bound, we state a general property to characterize these sketching matrices. This key property will be used in the later proof.  

\begin{definition}[Sign structure]\label{def:sign_structure}
For any sketching matrix, we say it has ``Sign structure'' if the following properties hold
\begin{itemize}
    \item $S_{k,i}\in \{\pm \frac{1}{\sqrt m} \}$, for all $k\in [m], i\in [n]$.
    \item $S_{k,i}$ and $S_{k,j}$ are independent for any $i\neq j$.
    \item $\E[S_{k,i}] = 0$ for all $k \in [m]$ and $i \in [n]$.
\end{itemize}
\end{definition}

\begin{lemma}\label{lem:satisfy_def}
Both
{\sf SRHT} and {\sf SRCT} satisfy Definition~\ref{def:sign_structure}.
\end{lemma}
\begin{proof}
It follows from the definitions of two sketching matrices directly.
\end{proof}

\begin{lemma}[{\sf SRHT} and {\sf SRCT}]
\label{lem:srht_implies_oce}

Let $S\in \R^{m\times n}$ be any sketching matrices that satisfy the Definition~\ref{def:sign_structure}. Then, we have
\begin{align*}
    \Pr[\max_{i\neq j}|\langle S_{*,i}, S_{*, j} \rangle |\geq \frac{\sqrt{\log(n/\delta)}}{\sqrt{m}}] \leq & ~ \Theta(\delta).
\end{align*}
\end{lemma}

\begin{proof}
Fix a pair of indices $i\neq j$ and we define $X \in \R^{m \times 2}$ as follows:
\begin{align*}
X : =\begin{bmatrix}
    S_{*,i} & S_{*,j}
\end{bmatrix} 
\end{align*}

The Gram matrix is $X^\top X=\sum_{k=1}^m G_k$, where
\begin{align*}
    G_k 
    = & ~ 
    \begin{bmatrix}
        S_{k,i} & S_{k,j}
    \end{bmatrix}^\top 
    \begin{bmatrix}
        S_{k,i} & S_{k,j}
    \end{bmatrix} \\
    = & ~ \begin{bmatrix}
    S_{k,i} \\
    S_{k,j}
    \end{bmatrix}
    \begin{bmatrix}
        S_{k,i} & S_{k,j}
    \end{bmatrix} \\
    = & ~ 
    \begin{bmatrix}
        S_{k,i}^2 & S_{k,i}S_{k,j} \\
        S_{k,i}S_{k,j} & S_{k,j}^2
    \end{bmatrix} \\
    = & ~ 
    \begin{bmatrix}
        \frac{1}{m} & S_{k,i}S_{k,j} \\
        S_{k,i}S_{k,j} & \frac{1}{m}
    \end{bmatrix}.
\end{align*}
where the first step follows from the definition of $G_k$, the second step follows from rewriting $\begin{bmatrix}
        S_{k,i} & S_{k,j}
    \end{bmatrix}^\top$, the third step follows from the definition of matrix multiplication, and the last step follows from $S_{k,i}^2 =1/m$ and $S_{k,j}^2=1/m$.

Note that $G_k$ has eigenvalues 0 and $\frac{2}{m}$, i.e.,
\begin{align*}
\lambda_1(G_k) =  ~ 2/m, ~~~
\lambda_2(G_k) =  ~ 0.
\end{align*}
Since $S_{k,i}$ and $S_{k,j}$ are independent Rademacher random variables, we have 
\begin{align*}
\E[S_{k,i} S_{k,j}] = \E[S_{k,i}] \cdot \E[S_{k,j}] = 0.
\end{align*}
Thus, we know
\begin{align}\label{eq:E_Gk}
\E[G_k]= \begin{bmatrix} 1/m & 0 \\ 0  & 1/m \end{bmatrix}.
\end{align}

Consequently, we have 
\begin{align*}
\E[X^\top X]
= & ~ \E[ \sum_{k=1}^m G_k ] 
= ~ m \cdot \E[G_k] \\
= & ~ m \cdot \begin{bmatrix} 1/m & 0 \\ 0  & 1/m \end{bmatrix} \\
= & ~ \begin{bmatrix} 1 & 0 \\ 0  & 1 \end{bmatrix},
\end{align*}
where the first step follows from the definition of $X^\top X$, the second step follows from the fact that $\E[ca] = c\E[a]$ for a constant $c$, the third step follows from Eq.~\eqref{eq:E_Gk}, and the last step follows from simple algebra.

Let $\lambda_i(X^\top X)$ be the $i$-th eigenvalue of $X^\top X \in \R^{2 \times 2}$. By matrix Chernoff bound (Lemma~\ref{lem:chernoff_bound} with $B=2/m$), for any $t > 0$, we have 
\begin{align*}
    \Pr[ \forall i \in [2], |\lambda_{i}(X^\top X)-1| \geq t] \leq & ~ 4\exp(-t^2 m/2)
\end{align*}

This means with probability at least $1-4\exp(-t^2m/2)$, the eigenvalues of $X^\top X$ are between $[1-t,1+t]$ and consequently, the eigenvalues of $X^\top X-I_2$ are between $[-t, t]$. 
Let us choose $t=O(\frac{\sqrt{\log(n/\delta)}}{\sqrt{m}})$, we have
\begin{align*}
    \Pr[\|X^\top X-I_2\| \geq C\cdot \frac{\sqrt{\log(n/\delta)}}{\sqrt{m}}] \leq & ~ \frac{\delta}{n^2}.
\end{align*}
The proof can be wrapped up by noting that 
\begin{align*}
    X^\top X-I_2 = & ~ \begin{bmatrix}
       0 & \langle S_{*,i}, S_{*,j}\rangle \\
        \langle S_{*,i},S_{*,j}\rangle & 0 
    \end{bmatrix},
\end{align*}
the spectral norm of this matrix is $|\langle S_{*,i}, S_{*,j}\rangle |$ and union bound over all $n^2$ pairs of columns, we have
\begin{align*}
    \Pr[\max_{i\neq j} |\langle S_{*,i}, S_{*,j} |\geq C\cdot \frac{\sqrt{\log(n/\delta)}}{\sqrt{m}}] \leq \delta. 
\end{align*}
\end{proof}

\subsection{Column norm bound for {\sf SRHT} and {\sf SRCT}}

In this section, we prove the column norm bound for {\sf SRHT} and {\sf SRCT}. In particular, their columns are unit vectors. In Appendix~\ref{sec:gaussian_ams}, we prove for random Gaussian matrix, the squared column norm is $\chi^2_m$ random vriable that concentrates around 1 with high probability.

\begin{lemma}[{\sf SRHT} and {\sf SRCT}]\label{lem:S*i22}
Let $S\in \R^{m\times n}$ be an {\sf SRHT} matrix or {\sf SRCT} matrix. 

Then, for any $i\in [n]$, we have
$
    \|S_{*,i}\|_2^2 =  ~ 1.
$
\end{lemma}

\begin{proof}
The proof directly follows from the definition.

For {\sf SRHT}, recall $S=PHD$, the column norm of $H$ is $\sqrt{n}$, and $D$ is a random sign that does not change the norm. The matrix $P$ subsamples $m$ rows and rescale each entry by $\sqrt{\frac{1}{m}}$. The (squared) column norm is then 1.

For {\sf SRCT}, the column norm of $G$ is $\sqrt{n}$ as well. Thus, by the same argument, {\sf SRCT} has its column vectors being units.
\end{proof}

\section{Put things together}
\label{sec:put_together}

Now, we're ready to present the proof for Theorem~\ref{thm:standard_version}.
\begin{proof}[Proof of Theorem~\ref{thm:standard_version}]

Using Lemma~\ref{lem:srht_implies_ose} (it shows {\sf SRHT} gives {\sf OSE}), we know if $m \geq d \log^2(n/\delta)$, it gives $(O(1), \delta, n,d)$-OSE.

Using Lemma~\ref{lem:srht_implies_oce} (it shows {\sf SRHT} gives {\sf OCE}), we know $\beta  = O(\log^{1.5}(n/\delta))$.

Using Lemma~\ref{lem:core_lemma} (it shows {\sf OSE} + {\sf OCE} implies our result), we need to choose 
\begin{align*}
m \geq \epsilon^{-2} d \beta^2 \geq \epsilon^{-2} d \log^{3}(n/\delta) 
\end{align*} 

Combining the above equation together, we have
\begin{align*}
m \geq ~ d \log^2(n/\delta) + \epsilon^{-2} d \log^{3}(n/\delta) 
\geq ~ \epsilon^{-2} d \log^3(n/\delta). 
\end{align*}
\end{proof}

\section{Conclusion}\label{sec:conclusion}

In this paper, we study the sketching-based regression algorithm with an $\ell_\infty$ guarantee. We show that {\sf SRHT} with $m=\epsilon^{-2}d\log^3(n/\delta)$ rows provides the desired $\ell_\infty$ guarantee solution, improving upon the $\epsilon^{-2}d^{1+\gamma}$ rows for $\gamma=\sqrt{\frac{\log\log n}{\log d}}$ of~\cite{psw17}. This is nearly-optimal up to logarithmic factors. Our proof adapts the oblivious coordinate-wise embedding property introduced in~\cite{sy21} in a novel way. We also greatly extends the reach of $\ell_\infty$ guarantee to degree-2 Kronecker product regression via {\sf TensorSRHT} matrix.

In addition, we introduce the {\sf SRCT} and {\sf TensorSRCT} matrices. These matrices can be applied in a fashion similar to {\sf SRHT}, and they have similar {\sf OCE} behaviors as {\sf SRHT}.

Our result provides an elegant way to integrate fast, sketching-based regression solver for optimization process, in particular second-order methods. The regression problem per iteration can be solved in time nearly-linear in the input size, and the $\ell_\infty$ guarantee comes in handy when analyzing convergence with approximate step. It also gives improved generalization bound on approximate regression via {\sf SRHT}~\cite{psw17}.

\ifdefined\isarxiv
\else
\bibliography{ref}
\bibliographystyle{icml2023}

\fi

\newpage
\onecolumn
\appendix

\section*{Appendix}
\paragraph{Roadmap.} 

In Section~\ref{sec:basic_tools}, we introduce the fundamental definitions and properties that we will use in Appendix. In Section~\ref{sec:kronecker_product_regression}, we analyze and develop the $\ell_\infty$ guarantee of Kronecker product regressions. In Section~\ref{sec:strong_JL_moment_property}, we introduce the Strong JL Moment Property and prove that both Circulant Transform and Tensor Circulant Transform satisfy this. In Section~\ref{sec:gaussian_ams}, we focus on studying AMS, random Gaussian, and SRHT and show that the inner product is bounded on any pair of different columns of AMS, random Gaussian, and SRHT--dense sketching matrices.

\section{Tools for matrices and probability}
\label{sec:basic_tools}

For matrix $A_{1} \in \R^{n_1 \times d_1}$ and $A_2 \in \R^{n_2 \times d_2}$, we use $A_1 \otimes A_2 \in \R^{n_1 n_2 \times d_1 d_2}$ to denote the matrix that $(i_1-1) \cdot (n_2) + i_2, (j_1-1) d_2 + j_2$ -th entry is $(A_1)_{i_1,j_1} \cdot (A_2)_{i_2,j_2}$.

\begin{lemma}[Markov's inequality]\label{lem:markov_inequality}
If $X$ is a non-negative random variable and $a>0$. Then we have
\begin{align*}
\Pr[ X \geq a ] \leq \E[X] / a. 
\end{align*}
\end{lemma}

\begin{definition}[Sub-exponential distribution (\cite{fkz11})]\label{def:sub-exponential_distribution}
    We say $X \in {\sf SubExp}(\sigma^2, \alpha)$ with parameters $\sigma > 0$, $\alpha > 0$ if:
    \begin{align*}
        \E[e^{\lambda X}] \leq \exp{(\lambda^2\sigma^2/2)}, \forall |\lambda| < 1/\alpha.
    \end{align*}
\end{definition}

\begin{lemma}[Tail bound for sub-exponential distribution (\cite{fkz11})]\label{lem:tail_bound}
    Let $X \in {\sf SubExp}(\sigma^2, \alpha)$ and $\E[X] = \mu$. 
    Then, 
    \begin{align*}
        \Pr[|X - \mu| \geq t] \leq \exp{(-0.5 \min\{t^2/\sigma^2, t/\alpha\})}.
    \end{align*}
\end{lemma}

\begin{claim}
\label{cla:rewrite_A_B_otimes_C_D}
    For every matrix $A \in \R^{n_1 \times n_2}, B \in \R^{n_2 \times n_3}, C \in \R^{d_1 \times d_2}, D \in \R^{d_2 \times d_3}$
    \begin{align*}
        (A \cdot B) \otimes (C \cdot D) = (A \otimes C) \cdot (B \otimes D).
    \end{align*}
\end{claim}
\section{Kronecker product regression with $\ell_\infty$ guarantee}
\label{sec:kronecker_product_regression}

In this section, we study the $\ell_\infty$ guarantee of Kronecker product regressions. Given two matrices $A_1, A_2\in \R^{n\times d}$ and a label vector $b\in \R^{n^2}$, the goal is to solve the regression $\arg\min_{x\in \R^{d^2}} \|(A_1\otimes A_2)x-b \|_2^2$. This problem can be easily generalized to product of $q$ matrices and fast, input-sparsity time algorithms have been studied in a line of works~\cite{dssw18,swz19,djs+19,ffg22,rsz22}.

\subsection{Main result}

\begin{theorem}[Tensor version of Theorem~\ref{thm:standard_version}]\label{thm:tensor_version}
Suppose $n \leq \exp(d)$ and matrix $A \in \R^{n^2 \times d^2}$ and vector $b \in \R^{n^2}$ are given, where $A = A_1 \otimes A_2$ for matrices $A_1, A_2 \in \R^{n \times d}$ and $b=b_1\otimes b_2$ for vectors $b_1,b_2\in \R^{n}$. Let $S \in \R^{m \times n^2}$ be a
\begin{itemize}
    \item tensor subsampled randomized Hadamard transform matrix ({\sf TensorSRHT}) with $m=\Theta(\epsilon^{-2}d^2\log^3(n/\delta))$ rows or
    \item tensor subsampled randomized circulant transform matrix ({\sf TensorSRCT}) with $m= \Theta(\epsilon^{-2} d^4 \log^3(n/\delta))$ rows. 
\end{itemize}

For 
\begin{align*}
    x' = \arg\min_{x \in \R^{d^2}} \| SA x - S b \|_2
\end{align*}
and 
\begin{align*}
    x^* = \arg\min_{x \in \R^{d^2}} \| A x - b \|_2,
\end{align*}
and any fixed $a \in \R^{d^2}$,
\begin{align*}
| \langle a, x^* \rangle - \langle a , x' \rangle | \leq \frac{ \epsilon } { d } \cdot\| a \|_2 \cdot \| A x^* - b \|_2 \cdot \| A^\dagger \|
\end{align*}
with probability $1-1/\poly(d)$.
\end{theorem}
\begin{proof}

Recall that we require $(O(1),\delta,d,n)$-{\sf OSE} and $\beta=O(\log^{1.5}(n/\delta))$-{\sf OCE} for it to give $\ell_\infty$ guarantee.

For {\sf OCE}, it follows from Lemma~\ref{lem:tensor_srht_implies_oce}.

For {\sf TensorSRHT}'s {\sf OSE}, it follows from Lemma~\ref{lem:tensor_srht_implies_ose} and for {\sf TensorSRCT}, it follows from Corollary~\ref{cor:tensor_srct_ose_1}. 
\end{proof}

\begin{remark}\label{rem:slightly_different_guarantee}
The slightly different guarantee follows from the small dimension becomes $d^2$ instead of $d$. Let us discuss the utility of using these sketching matrices for solving the regression. As discussed in Def.~\ref{def:tensor_srht} and~\ref{def:Tensor_CGT}, each column of $A_1\otimes A_2$ can be computed in $O(n\log n+m)$ time instead of $n^2$, thus the total running time of applying $S$ to $A$ is $O(nd^2\log n+\poly(d))$. Similarly, $Sb$ can be applied in time $O(n\log n+\poly(d))$. The regression can then be solved in $\wt O(nd^2+\poly(d))$ time. Prior works mainly focus on input-sparsity sketches~\cite{dssw18}, importance sampling~\cite{djs+19}, iterative method~\cite{ffg22} or more complicated sketches that scale well to $q$ products and in dynamic setting~\cite{rsz22}. To the best of our knowledge, this is the first $\ell_\infty$ guarantee for Kronecker product regression (with two matrices).
\end{remark}

\subsection{Oblivious coordinate-wise embedding for {\sf TensorSRHT} and {\sf TensorSRCT}}

\begin{lemma}[{\sf TensorSRHT} and {\sf TensorSRCT}, Tensor version of Lemma~\ref{lem:srht_implies_oce}]
\label{lem:tensor_srht_implies_oce}

Let $S\in \R^{m\times n}$ be {\sf TensorSRHT} or {\sf TensorSRCT}. Then, $S$ is an {\sf OCE} with parameter $\beta=\log^{1.5}(n/\delta)$.
\end{lemma}
\begin{proof}
To prove this result, we show that {\sf TensorSRHT} and {\sf TensorSRCT} satisfy Definition~\ref{def:sign_structure}.

For {\sf TensorSRHT}, recall $S=\frac{1}{\sqrt{m}} P(HD_1\times HD_2)$, since $H$ is Hadamard matrix and $D_1, D_2$ are just diagonal matrices with random signs. Thus, all entries of $HD_1\times HD_2$ are also in $\{\pm 1 \}$. As $P$ is a row sampling matrix and we rescale each entry by $\frac{1}{\sqrt{m}}$. Thus, each entry of $S$ is in $\{\pm \frac{1}{\sqrt{m}} \}$. For entries at the same row but two different columns $i, j$, if $i$ is generated from two columns disjoint from $j$, then it's clear then are independent. Otherwise, suppose $i$ is generated from columns $a, b$ and $j$ is generated from columns $a, c$ with $b\neq c$. Then it is again independent, as the sign is completely determined by signs of $b$ and $c$. Finally, we need to verify $\E[S_{k,i}]=0$, this is trivially true since product of two random signs is still a random sign. For {\sf TensorSRCT}, the argument is exactly the same.

Now that both of these matrices satisfy Definition~\ref{def:sign_structure}, we can use Lemma~\ref{lem:srht_implies_oce} to give a bound on pairwise inner product. The column norm bound is automatically satisfied by definition. Thus, we can invoke Lemma~\ref{lem:SRHT_AMS} to wrap up the proof.
\end{proof}
\section{{\sf SRCT} and {\sf TensorSRCT}: {\sf OSE} via strong JL moment property}
\label{sec:strong_JL_moment_property}

In this section, we prove that both {\sf SRCT} and {\sf TensorSRCT} are {\sf OSE}'s. We prove this property via the strong JL moment property~\cite{akk+20}. This gives a worse row count compared to that of {\sf SRHT} and {\sf TensorSRHT}. We believe that these two family of distributions should have similar row count for an {\sf OSE} and leave it as a major open problem to close the gap between these two distributions.

\subsection{Notations}

To make the notation less heavy, we will use $\|X\|_{L^{t}}$ for the $t$-th moment of a random variable $X$. This is formally defined below.
\begin{definition}[$t$-th moment]
\label{def:t_moment}
For every integer $t \geq 1$ and any random variable $X \in \R$, we write
\begin{align*}
    \|X\|_{L^{t}} = \left(\E \left[|X|^{t}\right] \right)^{1/t}
\end{align*}
\end{definition}
Note that 
\begin{align*}
    \|X+Y\|_{L^{t}} \leq \|X\|_{L^{t}} + \|Y\|_{L^{t}}
\end{align*} 
for any random variables $X$, $Y$ by the Minkowski inequality.

\subsection{Strong JL moment property}

We show that both {\sf SRCT} (see Definition~\ref{def:CGT}) and {\sf TensorSRCT} (see Definition~\ref{def:Tensor_CGT}) satisfy the so-called \emph{strong JL moment property}. Strong JL moment property is one of the core properties that can show the sketching matrix has subspace embedding property~\cite{s06}.

\begin{definition}[Strong JL moment property~\cite{akk+20}]
\label{def:strong_JL_moment}
For every $\eps, \delta \in [0,1]$, we say a distribution over random matrices $S \in \R^{m \times n}$ has the Strong $(\eps, \delta)$-JL Moment Property when
\begin{align*}
    \| \|Sx\|_{2}^{2} - 1 \|_{L^{t}} \leq \eps \sqrt{{t} / {\log(1/\delta)}}
\end{align*}
and
\begin{align*}
    \E \left[ \| Sx \|_{2}^{2} \right] = 1
\end{align*}
for all $x \in \R^{n}$, $\|x\|_{2} = 1$ and every integer $t \leq \log(1/\delta)$.
\end{definition}

Given a distribution with strong JL moment property, it is well-known that such distribution provides {\sf OSE}.

\begin{lemma}[Lemma 11 of~\cite{akk+20}]
\label{lem:strong_jl_to_ose}
Let $S\in \R^{m\times n}$ be a random matrix with $(\epsilon/d,\delta)$-strong JL moment property (Def.~\ref{def:strong_JL_moment}). Then, $S$ is also an $(\epsilon,\delta,d,n)$-{\sf OSE} (Def.~\ref{def:ose}).
\end{lemma}

To prove that {\sf SRCT} (see Definition~\ref{def:CGT}) and {\sf TensorSRCT} (see Definition~\ref{def:Tensor_CGT}) satisfy the strong JL moment property, we will do this by proving that a more general class of matrices satisfies the strong JL moment property. 

More precisely, let $k \in \Z_{>0}$ be a positive integer and 
\begin{align*}
    (D^{(i)})_{i \in [k]} \in \prod_{i \in [k]}\R^{n_{i} \times n_{i}}
\end{align*}
be independent matrices, each with diagonal entries given by independent Rademacher variables. 

Let $n = \prod_{i \in [k]} n_{i}$ and $P \in \{0,1\}^{m \times n}$ be a random sampling matrix in which each row contains exactly one uniformly distributed nonzero element which has value one. 

Then, we prove that the matrix 
\begin{align*}
    S = \frac{1}{\sqrt{m}}PG\cdot (D_{1} \otimes \dots \otimes D_{k})
\end{align*}
satisfies the strong JL moment property, where $G$ is $n \times n$ circulant matrix (see Definition~\ref{def:circulant_matrix}) generated by a random vector whose elements are Rademacher variables. 

If $k=1$, then $S$ is just a {\sf SRCT} (see Definition~\ref{def:CGT}). If $k=2$, then $M$ is a {\sf TensorSRCT} (see Definition~\ref{def:Tensor_CGT}).

In order to prove this result we need a couple of lemmas. The first lemma can be seen as a version of Khintchine's Inequality (see Lemma~\ref{lem:khintchine}) for higher order chaos.

\begin{lemma}[Lemma 19 in \cite{akk+20}]
\label{lem:bound_moment_of_sigma_dot_a}
    Let $t \geq 1$ and $k\in \mathbb{Z}_{>0}$. Let $(\sigma^{(i)})_{i \in [k]} \in \prod_{i\in [k]} \R^{n_i}$ be independent vectors each satisfying the Khintchine's inequality (see Lemma~\ref{lem:khintchine}):
    \begin{align*}
        \|\langle \sigma^{(i)}, x \rangle\|_{L^t} \leq C_t \|x\|_2
    \end{align*}
    for $t \geq 1$ and any vector $x \in \R^{d_i}$.

    Let $(a_{i_1,\ldots,i_k})_{i_1 \in [n_j], \ldots, i_k \in [n_k]}$ be a tensor in $\R^{n_1 \times \cdots \times n_k}$. Then,
    \begin{align*}
        \left\|\sum_{i_1 \in [n_1], \ldots, i_k \in [n_k]} (\prod_{j\in [k]} \sigma_{i_j}^{(j)}) a_{i_1,\ldots,i_k}\right\|_{L^t} \leq C_t^k (\sum_{i_1 \in [n_1], \ldots, i_k \in [n_k]}  a_{i_1,\ldots,i_k}^2)^\frac{1}{2},
    \end{align*}
    for $t \geq 1$. 
    
    Viewing $a \in \R^{n_1\times \ldots\times n_k}$ as a vector, then
    \begin{align*}
        \|\langle \sigma^{(1)} \otimes \cdots \otimes \sigma^{(k)}, a \rangle\|_{L^t} \leq C_t^k \|a\|_2,
    \end{align*}
    for $t \geq 1$.
\end{lemma}

\begin{proof}
    The proof will be by induction on $k$.

    {\bf Base case:} For $k = 1$, the result is by the assumption that the vectors satisfy Khintchine's inequality.

     {\bf Inductive case:} Assume that the result is true for every value up to $k - 1$. 
     
     Let 
      \begin{align}\label{eq:B_i_1_ldots_i_k_1}
          B_{i_1, \ldots, i_{k - 1}} = \sum_{i_k \in [d_k]} \sigma_{i_k}^{(k)} a_{i_1,\ldots,i_k}.
      \end{align}

      We then pull it out of the left hand term in the theorem:
      \begin{align*}
          \|\sum_{i_1 \in [n_1], \ldots, i_k \in [n_k]} (\prod_{j\in [k]} \sigma_{i_j}^{(j)}) a_{i_1,\ldots,i_k}\|_{L^t} 
          = & ~ \|\sum_{i_1 \in [n_1], \ldots, i_{k-1} \in [n_{k - 1}]} (\prod_{j\in [k - 1]} \sigma_{i_j}^{(j)}) B_{i_1,\ldots,i_{k - 1}}\|_{L^t}\\
          \leq & ~ C_t^{k - 1} \|(\sum_{i_1 \in [d_1], \ldots, i_{k-1} \in [n_{k - 1}]} B_{i_1,\ldots,i_{k - 1}}^2)^\frac{1}{2}\|_{L^t}\\
          = & ~ C_t^{k - 1} \|\sum_{i_1 \in [n_1], \ldots, i_{k-1} \in [n_{k - 1}]} B_{i_1,\ldots,i_{k - 1}}^2\|_{L^{t/2}}^\frac{1}{2}\\
          \leq & ~ C_t^{k - 1} (\sum_{i_1 \in [n_1], \ldots, i_{k-1} \in [n_{k - 1}]} \|B_{i_1,\ldots,i_{k - 1}}^2\|_{L^{t/2}})^\frac{1}{2}\\
          = & ~ C_t^{k - 1} (\sum_{i_1 \in [n_1], \ldots, i_{k-1} \in [n_{k - 1}]} \|B_{i_1,\ldots,i_{k - 1}}\|_{L^t}^2)^\frac{1}{2},
      \end{align*}
      where the first step follows from Eq.~\eqref{eq:B_i_1_ldots_i_k_1}, the second step follows from the inductive hypothesis, the third step follows from the definition of $\|\cdot\|$, the fourth step follows from  the triangle inequality, the fifth step follows from the definition of $\|\cdot\|$.

      It remains to bound
      \begin{align*}
          \|B_{i_1,\ldots,i_{k - 1}}\|_{L^t}^2 \leq C_t^2 \sum_{i_k \in [n_k]} a_{i_1,\ldots,i_k}^2
      \end{align*}
      by Khintchine’s inequality, which finishes the induction step and hence the proof.
\end{proof}

The next lemma we will be using is a type of Rosenthal inequality based on first principles. It mixes large and small moments of random variables in an intricate way. For completeness, we include a proof here.

\begin{lemma}[Properties of random variables with $t$-moment, Lemma 20 in~\cite{akk+20}]
\label{lem:t_moment}
There exists a universal constant $L$, such that, for $t \geq 1$ if $X_{1}, \dots, X_{k}$ are independent non-negative random variables with $t$-moment, then
\begin{align*}
    \|\sum_{i \in [k]} (X_{i} - \E[X_{i}])\|_{L^{t}} \leq L \cdot \big( \sqrt{t} \cdot \|\max_{i \in [k]}X_{i}\|^{1/2}_{L^{t}} \cdot (\sum_{i \in [k]}\E[X_{i}])^{1/2} + t \cdot \|\max_{i \in [k]} X_{i}\|_{L^{t}} \big).
\end{align*}
\end{lemma}

\begin{proof}
Throughout these calculations $L_{1}, L_{2}$ and $L_{3}$ will be universal constants.
\begin{align}\label{eq:X_i_EX_i}
    \|\sum_{i \in [k]} (X_{i} - \E[X_{i}]) \|_{L^{t}} 
    \leq & ~ L_{1} \|\sum_{i \in [k]} \sigma_{i} X_{i} \|_{L^{t}} \notag\\
    \leq & ~ L_{2} \sqrt{t} \cdot \| \sum_{i \in [k]} X^{2}_{i} \|^{1/2}_{L^{t/2}} \notag\\
    \leq & ~ L_{2} \sqrt{t} \cdot \| \max_{i \in [k]} X_{i} \cdot \sum_{i \in [k]} X_{i} \|^{1/2}_{L^{t/2}} \notag\\
    \leq & ~ L_{2} \sqrt{t} \cdot  \| \max_{i \in [k]} X_{i} \|^{1/2}_{L^{t}} \cdot \| \sum_{i \in [k]} X_{i} \|^{1/2}_{L^{t}} \notag\\
    \leq & ~ L_{2} \sqrt{t} \cdot \|\max_{i \in [k]} X_{i} \|^{1/2}_{L^{t}} \cdot  \Big( (\sum_{i \in [k]} \E [X_{i}])^{1/2} + L_{2} \| \sum_{i \in [k]} (X_{i} - \E[X_{i}]) \|^{1/2}_{L^{t}} \Big)
\end{align}
where the first step follows from symmetrization of $X_{i}$, the second step follows from Khintchine's inequality (see Lemma~\ref{lem:khintchine}), the third step follows from Non-negativity of $X_{i}$, the fourth step follows from Cauchy-Schwartz inequality, and the last step follows from the triangle inequality.

Now, let $A, B, C$ be defined as follows:
\begin{align*}
    C := \|\sum_{i \in [k]}(X_{i} - \E[X_{i}])\|^{1/2}_{L^{t}},
\end{align*}
\begin{align*}
    B := L_{2} (\sum_{i \in [k]}\E[X_{i}])^{1/2},
\end{align*}
 and 
\begin{align*}
    A := \sqrt{t} \| \max_{i \in [k]} X_{i} \|^{1/2}_{L^{t}}.
\end{align*} 
Then, by rewriting Eq.~\eqref{eq:X_i_EX_i}, we have 
\begin{align*}
    C^{2} \leq A(B+C).
\end{align*}

This implies $C$ is smaller than the largest of the roots of the quadratic. 

Solving this quadratic inequality gives 
\begin{align*}
    C^{2} \leq L_{3} (AB+A^{2}),
\end{align*}
which completes the proof.
\end{proof}

\subsection{{\sf SRCT} and {\sf TensorSRCT} satisfy strong JL moment property}
We can now prove that {\sf SRCT} (see Definition~\ref{def:CGT}) and {\sf TensorSRCT} (see Definition~\ref{def:Tensor_CGT}) have the strong JL moment property.

\begin{theorem}
\label{thm:PGD_strong_JL_moment}
There exists a universal constant $L$, such that, the following holds. 

Let $k \in \Z_{>0}$. Let $(D^{(i)})_{i \in [k]} \in \prod_{i \in [k]} \R^{n_{i} \times n_{i}}$ be independent diagonal matrices with independent Rademacher variables. 

We define $n := \prod_{i \in [k]} n_{i}$, $D := D_{1} \otimes D_{2} \otimes \dots \otimes D_{k} \in \R^{n \times n}$ and $G:=G_1\otimes \ldots \otimes G_k \in \R^{n\times n}$, where each $G_i\in \R^{n_i\times n_i}$ is a circulant matrix generated by an independent Rademacher random vector. Let $P \in \R^{m \times n}$ be a row sampling matrix that has exactly one nonzero per row. Let $S:=PGD$.

Let $x \in \R^{n}$ be any vector with $\|x\|_{2} = 1$ and $t \geq 1$. 

Then,
\begin{align*}
    \left\|\frac{1}{m} \|PGDx\|_{2}^{2} - 1\right\|_{L^{t}} \leq 
    L (\sqrt{\frac{tr^{k}}{m}} + \frac{tr^{k}}{m}),
\end{align*}
where $r = \max \{t, \log m\}$.

There exists a universal constant $C_0 > 1$, such that, by setting 
\begin{align*}
    m = \Omega(\eps^{-2} \log (1/\delta)  \cdot (C_0 \log(1/ (\eps \delta)) )^{k}),
\end{align*}
we get that $\frac{1}{\sqrt{m}}PGD$ has $(\eps, \delta)$-strong JL moment property.
\end{theorem}

\begin{proof}
Throughout the proof $C_{1}$, $C_{2}$ and $C_{3}$ will denote universal constants.

For every $i \in [m]$, we let $P_{i}$ be the random variable that says which coordinates the $i$-th row of $P$ samples. 

We define the random variable 
\begin{align*}
    Z_{i} := M_{i}x=G_{P_{i}}Dx.
\end{align*}

We note that since the variables $(P_{i})_{i \in [m]}$ are independent, so the variables $(Z_{i})_{i \in [m]}$ are conditionally independent given $D$, that is, if we fix $D$, then $(Z_{i})_{i \in [m]}$ are independent.

Then, we could get the following inequality:
\begin{align*}
    & ~ \|\frac{1}{m}\sum_{i \in [m]}Z^{2}_{i} - 1\|_{L^{t}} \\
    = & ~ \| (\E[(\frac{1}{m} \sum_{i \in [m]}Z^{2}_{i} - 1) ~\big|~ D])^{1/t} \|_{L^{t}} \\
    \leq & ~ C_{1} \| \frac{\sqrt{t}}{m} \cdot (\E[(\max_{i \in [m]}Z^{2}_{i}) ~\big|~ D])^{1/(2t)} \cdot (\sum_{i \in [m]} \E[Z^{2}_{i} ~\big|~ D])^{1/2} + \frac{t}{m} \cdot (\E[(\max_{i \in [m]} Z^{2}_{i})^{t} ~\big|~ D])^{1/t} \|_{L^{t}}   \\
    \leq & ~ C_{1} \frac{\sqrt{t}}{m} \cdot \| (\E[(\max_{i \in [m]}Z^{2}_{i}) ~\big|~ D])^{1/(2t)} \cdot (\sum_{i \in [m]} \E[Z^{2}_{i} ~\big|~ D])^{1/2} \|_{L^{t}} + C_{1} \frac{t}{m} \cdot \|\max_{i \in [m]} Z^{2}_{i} \|_{L^{t}} \\
    \leq & ~ C_{1} \frac{\sqrt{t}}{m} \cdot \| \max_{i \in [m]} Z^{2}_{i} \|^{1/2}_{L^{t}} \cdot \| \sum_{i \in [m]} \E[Z^{2}_{i} ~\big|~ D] \|^{1/2}_{L^{t}} + C_{1} \frac{t}{m} \cdot \| \max_{i \in [m]} Z^{2}_{i} \|_{L^{t}}
\end{align*}
where the first step follows from Definition~\ref{def:t_moment}, the second step follows from Lemma~\ref{lem:t_moment}, the third step follows from triangle inequality, and the last step follows from Cauchy-Schwartz inequality.

Note that each row of $G$ is generated by taking the tensor product of independent Rademacher random vectors, we thus can view the row vector itself as a length $n$ Rademacher random vector. Thus,
\begin{align}\label{eq:x22}
    \E[Z^{2}_{i} | D] 
    = & ~  \sum_{\sigma \in \{-1,+1\}^n} p_{\sigma} \cdot (\langle x, \sigma \rangle)^2 \notag  \\
    = & ~ \frac{(x_{1} + x_{2} + \cdots + x_n )^{2}}{2^n} + \frac{(x_{1} + x_{2} + \cdots - x_n )^{2}}{ 2^n } + \dots + \frac{(- x_{1} - x_{2} - \cdots - x_n )^{2}}{2^n } \notag \\
    = & ~ x^{2}_{1} + x^{2}_{2} + \cdots + x^{2}_n \notag \\
    = & ~ \|x\|^{2}_{2},
\end{align}
where the first step follows from the definition of the expected value, $\E[Z^{2}_{i} | D]$, the second step follows from expanding all the $2^n$ possibilities, the third step follows from simple algebra, and the last step follows from the definition of $\|\cdot\|_2^2$.

We could get that
\begin{align*}
    \sum_{i \in [m]} \E[Z_{i}^{2} ~\big|~ D] 
    = & ~ \sum_{i \in [m]} \|x\|_{2}^{2} \\
    = & ~ m,
\end{align*}
where the first step follows from Eq.~\eqref{eq:x22} and the second step follows from $\|x\|^{2}_{2} = 1$.

To bound $\|\max_{i \in [m]} Z^{2}_{i}\|_{L^{t}}$, we could show
\begin{align*}
    \|Z^{2}_{i}\|_{L^{r}} 
    = & ~ \|G_{P^{i}}Dx\|^{2}_{L^{2r}} \\
    = & ~ \|Dx\|^{2}_{L^{2r}} \\
    \leq & ~ r^{k} \|x\|^{2}_{2}.
\end{align*}
where the first step follows from the definition of $Z_{i}$, the second step follows from each row of $G$ is independent Rademacher vector, therefore $\E[G^\top G]=I$, and the last step follows from Lemma~\ref{lem:bound_moment_of_sigma_dot_a}.

We then bound the maximum using a sufficiently high-order sum:
\begin{align*}
    \|\max_{i \in [m]} Z^{2}_{i} \|_{L^{t}} 
    \leq & ~ \|\max_{i \in [m]} Z^{2}_{i} \|_{L^{r}} \\
    \leq & ~ (\sum_{i \in [m]} \|Z^{2}_{i}\|^{r}_{L^{r}})^{1/r} \\
    \leq & ~ m^{1/r} r^{k} \|x\|_{2}^{2} \leq er^{k},
\end{align*}
where the first step follows from Definition~\ref{def:t_moment}, the second step follows from $Z^{2}_{i}$ is non-negative, and the last step follows from $r \geq \log m$. 

This gives us that
\begin{align}\label{eq:norm_L_t}
    \|\frac{1}{m} \sum_{i \in [m]} Z^{2}_{i} - \|x\|_{2}^{2} \|_{L^{t}} \leq C_{2} \sqrt{\frac{tr^{k}}{m}} + C_{2} \frac{tr^{k}}{m}
\end{align}
which finishes the first part of the proof.

We want to choose $m$ as follows 
\begin{align*}
    m= 16 C^{2}_{2}\eps^{-2} \cdot \log(1/\delta) \cdot (C_{3}\log(1/(\delta \eps)))^{k}.
\end{align*}

According to the above choice of $m$, we know following condition for $r$ is holding
\begin{align*}
    r \leq C_{3} \log (1/(\delta \eps)).
\end{align*}

Hence, 
\begin{align*}
    m \geq 16 C^{2}_{2} \eps^{-2} \cdot \log(1/\delta) \cdot r^{k}.
\end{align*}

For all $1 \leq t \leq \log (1/\delta)$, we then get that
\begin{align*}
    \| \|PGDx\|_{2}^{2} - 1\|_{L^{t}} 
    \leq & ~  C_{2} \sqrt{\frac{tr^{k}}{m}} + C_{2} \frac{tr^{k}}{m} \\
    \leq & ~ C_{2} (\frac{tr^{k}}{16 C_{2}^{2}\eps^{-2}\log(1/\delta)r^{k}})^{1/2} + C_{2} \frac{tr^{k}}{ 16 C^{2}_{2} \eps^{-2} \log (1/\delta)r^{k}} \\
    \leq & ~ 0.5 \epsilon \sqrt{t/\log(1/\delta)} + 0.5  \epsilon^2 t / \log(1/\delta) \\
    \leq & ~ \epsilon \sqrt{{t}/{\log(1 / \delta)}}.
\end{align*}
 where the first step follows from Eq.~\eqref{eq:norm_L_t}, and the second step follows from choice of $m$, the third step follows from simple algebra, and the last step follows from $\epsilon^2 \leq \epsilon$ and $t/\log(1/\delta) \sqrt{t/\log(1/\delta)}$ (since $t/\log(1/\delta) \in (0,1)$) .

This finishes the proof.
\end{proof}

As two corollaries, we have {\sf SRCT} and {\sf TensorSRCT} are {\sf OSE}'s with $d^2$ rows, instead of $d$ rows.

\begin{corollary}[{\sf SRCT} is an {\sf OSE}]
\label{cor:srct_imply_ose}
Let $S\in \R^{m\times n}$ be an {\sf SRCT} matrix with $m=\Theta(\epsilon^{-2}d^2\log^2(n/\epsilon\delta))$, then $S$ is an $(\epsilon,\delta,d,n)$-{\sf OSE}.
\end{corollary}

\begin{proof}
The proof follows from combining Lemma~\ref{lem:strong_jl_to_ose} and Theorem~\ref{thm:PGD_strong_JL_moment} with $k=1$.
\end{proof}

\begin{corollary}[{\sf TensorSRCT} is an {\sf OSE}]
\label{cor:tensor_srct_ose_1}
Let $S\in \R^{m\times n}$ be a {\sf TensorSRCT} matrix with $m=\Theta(\epsilon^{-2}d^2\log^3(n/\epsilon\delta))$, then $S$ is an $(\epsilon,\delta,d,n^2)$-{\sf OSE}.
\end{corollary}

\begin{proof}
The proof follows from combining Lemma~\ref{lem:strong_jl_to_ose} and Theorem~\ref{thm:PGD_strong_JL_moment} with $k=2$.
\end{proof}

\section{Gaussian and AMS}
\label{sec:gaussian_ams}

In this section, we prove that both random Gaussian matrices and AMS matrices satisfy {\sf OCE} with good parameter $\beta$. Combining with the fact that they are {\sf OSE}'s, one can derive $\ell_\infty$ guarantee for them.

\subsection{{\sf OSE} property of random Gaussian and AMS}

The {\sf OSE} property for these two distributions are folklore. For a proof for them, see, e.g.,~\cite{w14}.

\begin{lemma}
\label{lem:gaussian_imply_ose}
Let $S$ be a random Gaussian matrix defined in Def.~\ref{def:random_Gaussian}. If $m=\Theta(\epsilon^{-2}(d+\log(d/\delta)))$, then $S$ is an $(\epsilon,\delta,d,n)$-{\sf OSE}.
\end{lemma}

\begin{lemma}
\label{lem:ams_imply_ose}
Let $S$ be an AMS matrix defined in Def.~\ref{def:AMS}. If $m=\Theta(\epsilon^{-2}d\log^2(n/\delta))$, then $S$ is an $(\epsilon,\delta,d,n)$-{\sf OSE}.
\end{lemma}

\subsection{{\sf OCE} property of random Gaussian and AMS}

In this section, we prove the {\sf OCE} property of random Gaussian and AMS. We start with the pairwise inner product bound for these two distributions. For column norm bound, AMS has unit columns and we will prove for random Gaussian.

\begin{lemma}[Gaussian pairwise inner product bound, Lemma B.18 in \cite{sy21}]\label{lem:gaussian_S_i_S_j}
    Let $S \in \R^{m \times n}$ be a random Gaussian matrix (Definition~\ref{def:random_Gaussian}). 
    
    Then, we have:
    \begin{align*}
        \Pr[\max_{i\not= j} | \langle S_{*,i},  S_{*,j} \rangle | \geq C\cdot\frac{\sqrt{\log{(n/\delta)}}}{\sqrt{m}}] \leq \Theta(\delta).
    \end{align*}
\end{lemma}

\begin{proof}
    Note for $i \not= j$, $S_{*,i},  S_{*,j} \sim \mathcal{N}(0,\frac{1}{m} I_m)$ are two independent Gaussian vectors. Let $z_k = S_{k,i} S_{k,j}$ and $z = \langle S_{*,i},  S_{*,j} \rangle $. 
    
    Then, we have for any $|\lambda| \leq m/2$,
    \begin{align*}
        \E[e^{\lambda z_k}] = \frac{1}{\sqrt{1 - \lambda^2/m^2}} \leq \exp{(\lambda^2 / m^2)},
    \end{align*}
    where the first step follows from $z_k = \frac{1}{4} (S_{k,i} + S_{k,j})^2 +  \frac{1}{4} (S_{k,i} - S_{k,j})^2 = \frac{m}{2} (Q_1 - Q_2)$ where $Q_1,Q_2 \sim \chi_1^2$, and $\E[e^{\lambda Q}] = \frac{1}{\sqrt{1 - 2\lambda}}$ for any $Q\sim \chi_1^2$.

    This implies $z_k \in {\sf SubExp}(2/m^2, 2/m)$ is a sub-exponential random variable. Here ${\sf SubExp}$ is the shorthand of sub-exponential random variable.
    
    Thus, we have 
    \begin{align*}
        z = \sum_{k = 1}^m z_k \in {\sf SubExp}(2/m, 2/m),
    \end{align*}
    by sub-exponential concentration Lemma~\ref{lem:tail_bound}, we have
    \begin{align*}
        \Pr[|z| \geq t] \leq 2\exp{(-mt^2/4)}
    \end{align*}
    for $0 < t < 1$. Picking $t = \sqrt{\log{(n^2/\delta)/m}}$, we have
    \begin{align*}
        \Pr[|\langle S_{*,i},  S_{*,j} \rangle| \geq C\cdot \frac{\sqrt{\log{(n/\delta)}}}{\sqrt{m}}] \leq \delta/n^2.
    \end{align*}
    Taking the union bound over all $(i, j) \in [n] \times [n]$ and $i \not= j$, we complete the proof.
\end{proof}

\begin{lemma}[AMS pairwise inner product bound]\label{lem:AMS}
    Let $S \in \R^{m \times n}$ be an AMS matrix (Definition~\ref{def:AMS}. Let $\{\sigma_i\}_{i \in [n]}$ be independent Rademacher random variables and $\overline{S} \in \R^{m \times n}$ with $\overline{S}_{*, i} = \sigma_i S_{*, i}$, $\forall i \in [n]$. 
    
    Then, we have:
    \begin{align*}
        \Pr[\max_{i \not= j} |\langle \overline{S}_{*, i}, \overline{S}_{*, j}\rangle | \geq \frac{\sqrt{\log{(n/\delta)}}}{\sqrt{m}}] \leq \Theta(\delta)
    \end{align*}
\end{lemma}

\begin{proof}
    Note for any fixed $i \not= j$, $\overline{S}_{*, i}$ and $\overline{S}_{*, j}$ are independent. By Hoeffding inequality (Lemma~\ref{lem:hoeffding}), 
    we have
    \begin{align*}
        & ~ \Pr[|\langle \overline{S}_{*, i}, \overline{S}_{*, j} \rangle| \geq t] \\
        \leq & ~ 2 \exp{(-\frac{2t^2}{\sum_{i = 1}^m (\frac{1}{m} - (-\frac{1}{m}))^2})} \\
        \leq & ~ 2e^{-t^2m/2},
    \end{align*}
    where the second step follows from simple algebra ($m \cdot 1/m^2 = 1/m$).

    Choosing $t = \sqrt{2\log{(2n^2/\delta)}}/\sqrt{m}$, we have
    \begin{align*}
        \Pr[|\langle \overline{S}_{*, i}, \overline{S}_{*, j} \rangle| \geq \sqrt{2\log{(2n^2/\delta)}}/\sqrt{m}] \leq \frac{\delta}{n^2},
    \end{align*}
    union bound over all $n^2$ pairs of columns gives the desired result.
\end{proof}

\begin{lemma}[Gaussian column norm bound]\label{lem:gaussian_S_i_S_i}
Let $S\in \R^{m\times n}$ be a random Gaussian matrix. 

Then, for any $i\in [n]$, we have
\begin{align*}
    \Pr[|\|S_{*,i}\|_2^2-1 |\geq \frac{\sqrt{\log(n/\delta)}}{\sqrt{m}}] \leq & ~ \Theta(\delta).
\end{align*}
\end{lemma}

\begin{proof}
For any column $S_{*,i}$, note that $\|S_{*,i}\|_2^2\sim \chi_m^2$, each one with zero mean and variance $\frac{1}{m}$. 

By Lemma~\ref{lem:lm}, we have
\begin{align*}
    \Pr[|\|S_{*,i}\|_2^2-1 |\geq 2\frac{\sqrt{t}}{\sqrt{m}}] \leq & ~ 2\exp(-t).
\end{align*}
Setting $t=\log(n/\delta)$, we have
\begin{align*}
    \Pr[|\|S_{*,i}\|_2^2-1 |\geq C\cdot \frac{\sqrt{\log(n/\delta)}}{\sqrt{m}}] \leq & ~ \delta/n,
\end{align*}
the proof is concluded by union bounding over all $n$ columns.
\end{proof}

We conclude random Gaussian and AMS are {\sf OCE}'s.

\begin{lemma}[Gaussian {\sf OCE}]
\label{lem:gaussian_imply_oce}
Let $S\in \R^{m\times n}$ be a random Gaussian matrix, then $S$ is a $(\log^{1.5}(n/\delta),\delta,n)$-{\sf OCE}.
\end{lemma}

\begin{proof}
By Lemma~\ref{lem:gaussian_S_i_S_j} and Lemma~\ref{lem:gaussian_S_i_S_i}, we know both pairwise inner product bound and column norm bound hold and thus, by Lemma~\ref{lem:SRHT_AMS}, $S$ satisfies the desired {\sf OCE} property.
\end{proof}

\begin{lemma}[AMS {\sf OCE}]
\label{lem:ams_imply_oce}
Let $S\in \R^{m\times n}$ be an AMS matrix, then $S$ is a $(\log^{1.5}(n/\delta),\delta,n)$-{\sf OCE}.
\end{lemma}

\begin{proof}
The proof is similar to Lemma~\ref{lem:gaussian_imply_oce}. The column norm bound follows from definition.
\end{proof}

\ifdefined\isarxiv
\bibliographystyle{alpha}
\bibliography{ref}
\else

\fi




\end{document}